\documentclass{elsarticle}
\usepackage{amsmath,amsfonts,amsthm}
\usepackage{amssymb,latexsym}

\theoremstyle{plain}
\newtheorem{theorem}{Theorem}[section]
\newtheorem{lemma}[theorem]{Lemma}

\theoremstyle{definition}
\newtheorem{mydefn}[theorem]{Definition}
\newtheorem{example}[theorem]{Example}

\usepackage{mathrsfs}
\usepackage{amsfonts}
\usepackage{graphicx}
\usepackage{slashed}

\newcommand{\action}[1]{\stackrel{#1}{\rightarrow}}
\newcommand{\faction}[1]{\stackrel{#1}{\rightarrow}_F}
\newcommand{\actions}[1]{\stackrel{#1}{\Rightarrow}}
\newcommand{\sactions}[1]{\stackrel{#1}{\Rightarrow}|}
\newcommand{\fsactions}[1]{\stackrel{#1}{\Rightarrow}_F|}
\newcommand{\SRS}{\underset{\thicksim}{\sqsubset}_{RS}}
\newcommand{\RS}{\sqsubseteq_{RS}}

\begin{document}

\author[nuaa]{Yan Zhang}
%\ead{zhyi812@163.com}

\author[nuaa]{Zhaohui Zhu\corref{cor}}
%\ead{zhaohui@nuaa.com}

\author[nanshen]{Jinjin Zhang}
%\ead{zhangjinjin@163.com}

\cortext[cor]{Corresponding author. Email: zhaohui@nuaa.edu.cn (Zhaohui Zhu).}

\address[nuaa]{College of Computer Science, Nanjing University of Aeronautics and Astronautics, Nanjing, P.R. China, 210016}
\address[nanshen]{College of Information Science, Nanjing Audit University, Nanjing, P.R. China, 211815}

\title{On the greatest solution of equations in $\text{CLL}_R$ \tnoteref{t1}}
\tnotetext[t1]{This work received financial support of the National Natural Science of China (No. 60973045) and NSF of the Jiangsu Higher Education Institutions (No. 13KJB520012)}
\date{\today}
\begin{abstract}
It is shown that, for any equation $X=_{RS} t_X$  in the LLTS-oriented process calculus $\text{CLL}_R$, if $X$ is strongly guarded in $t_X$, then the recursive term $\langle X|X=t_X \rangle$ is the greatest solution of this equation w.r.t L\"{u}ttgen and Vogler's ready simulation.
\end{abstract}


\journal{Information processing letters}

\maketitle
%
%\begin{thanks}
%  hello,don't
%\end{thanks}

\section{Introduction}

   The notion of logic labelled transition system (LLTS for short), proposed by L\"{u}ttgen and Vogler, provides a framework to combine operational and logical styles of specification [2,3,4].
   Recently, inspired by this work, we propose an LLTS-oriented process calculus $\text{CLL}_R$, and establish the uniqueness of solutions of equations in $\text{CLL}_R$ under a certain circumstance [5].
   This note considers solutions of equations in $\text{CLL}_R$ furtherly.
   Firstly, through giving an example, it will be shown that, without the assumption that $X$ does not occur in the scope of any conjunction in $t$, an equation $X=_{RS}t$ may have more than one consistent solution.
   Secondly, under the hypothesis that $X$ is strongly guarded in a given term $t$, it will be shown that the process $\langle X|X=t\rangle$ is the greatest solution of the equation $X=_{RS}t$.
This result reveals that $\langle X|X=t \rangle$ captures the loosest specification satisfying the equation $X=_{RS}t$ whenever $X$ is strongly guarded in $t$.
    The rest of this note is organized as follows.
    The next section recalls some related notions and results.
   The main result will be given in Section~3.

\section{Preliminaries}

%We will fix some notations and terminologies, and recall some results.

%\subsection{Logic LTS and ready simulation}

%In this subsection, we introduce some useful notations and recall the definition of Logic LTS and ready simulation \cite{Luttgen11}.
This section will recall a number of related notions and results. Given space limitation, we only list these ones. For details see [2,3,4,5].
We begin with recalling the notion of LLTS.
Let $Act$ be the set of visible action names ranged over by $a$, $b$, etc., and let $Act_{\tau}$ denote $Act \cup \{\tau\}$ ranged over by $\alpha$ and $\beta$, where $\tau$ represents invisible actions.
A labelled transition system with predicate is a quadruple $(P,Act_{\tau},\rightarrow, F)$, where $P$ is a set of states, $\rightarrow \subseteq P\times Act_{\tau}\times P$ is the transition relation and $F\subseteq P$.
As usual, we write $p \stackrel{\alpha}{\rightarrow}$ (or, $p \not \stackrel{\alpha}{\rightarrow}$) if $\exists q\in P.p\stackrel{\alpha}{\rightarrow}q$ ($\nexists q\in P.p  \stackrel{\alpha}{\rightarrow}q$, resp.).
The ready set $\{\alpha \in Act_{\tau}:p \stackrel{\alpha}{\rightarrow}\}$ of a given state $p$ is denoted by $\mathcal{I}(p)$.
A state $p$ is stable if $p \not\stackrel{\tau}{\rightarrow}$.
Some useful decorated transition relations are listed below:

(1) $p \stackrel{\alpha}{\rightarrow}_F q$ iff $p \stackrel{\alpha}{\rightarrow} q$ and $p,q\notin F$;
(2) $p \stackrel{\epsilon}{\Rightarrow}q$ iff $p (\stackrel{\tau}{\rightarrow})^* q$, where $(\stackrel{\tau}{\rightarrow})^* $ is the transitive and reflexive closure of $\stackrel{\tau}{\rightarrow}$;
(3) $p \stackrel{\alpha}{\Rightarrow}q$ iff $\exists r,s\in P.p \stackrel{\epsilon}{\Rightarrow} r \stackrel{\alpha}{\rightarrow}s \stackrel{\epsilon}{\Rightarrow} q$;
(4) $p \stackrel{\gamma}{\Rightarrow}|q$ iff $p \stackrel{\gamma}{\Rightarrow}q \not\stackrel{\tau}{\rightarrow}$ with $\gamma \in Act_{\tau}\cup \{\epsilon\}$;
(5) $p\stackrel{\epsilon }{\Rightarrow }_Fq$ iff there exists a sequence of $\tau$-transitions from $p$ to $q$ such that all states along this sequence, including $p$ and $q$, are not in $F$; the decorated transition $p \stackrel{\alpha }{\Rightarrow }_Fq$ may be defined similarly;
(6) $p \stackrel{\gamma}{\Rightarrow}_F|q$ iff $p \stackrel{\gamma}{\Rightarrow}_F q \not\stackrel{\tau}{\rightarrow}$ with $\gamma \in Act_{\tau} \cup \{\epsilon\}$.

% Notice that
%the notation $p
%\stackrel{\gamma }{\Longrightarrow }\mspace{-8mu}|q$ in \cite{Luttgen10,Luttgen11}
%has the same meaning as $p\stackrel{\gamma }{
%\Rightarrow }_F|q$ in this paper, while $p\stackrel{\gamma }{\Rightarrow }|q $ in this paper does not involve any requirement on $F$-predicate.

\begin{mydefn}[{[3]}]\label{D:LLTS}
 An LTS $(P,Act_{\tau},\rightarrow,F)$ is an LLTS, if, for each $p \in P$,
\noindent\textbf{(LTS1) }$p \in F$ if $\exists\alpha\in \mathcal{I}(p)\forall q\in P(p \stackrel{\alpha}{\rightarrow}q \;\text{implies}\; q\in F)$;
\noindent\textbf{(LTS2)} $p \in F$ if $\nexists q\in P.p \stackrel{\epsilon}{\Rightarrow}_F|q$.
An LLTS $(P,Act_{\tau},\rightarrow,F)$ is {$\tau$}-pure if, for each $p \in P$, $p\stackrel{\tau}{\rightarrow}$ implies $\nexists a\in Act.\;p\stackrel{a}{\rightarrow}$.
\end{mydefn}

%Any state $p$ in a $\tau$-pure LTS represents either an external or internal choice between its outgoing transitions.
%$\tau$-purity is a technical constraints on LLTS.

Compared with usual LTSs, one distinctive feature of LLTS is that it involves consideration of inconsistencies.
The motivation behind such consideration lies in dealing with inconsistencies caused by conjunctive composition.
 The predicate $F$ in LLTS is used to denote the set of all inconsistent states.
 %with empty behaviour that cannot be implemented or reached at run-time.
%In the sequel, we shall use the phrase ``{\it inconsistency} {\it predicate}'' to refer to $F$.
The condition (LTS1) formalizes the backward propagation of inconsistencies, and (LTS2) captures
the intuition that divergence should be viewed as catastrophic.
%For more motivation behind such notion, we refer the reader to [2,3].
A variant of the usual notion of weak ready simulation is recalled below, which
is adopted to capture the refinement relation between processes in [3,4].
%It has been proven that such kind of ready simulation is the largest precongruence w.r.t parallel composition and conjunction which satisfies  the desired property that  an inconsistent specification can only be refined by inconsistent ones (see Theorem 21 in \cite{Luttgen10}).

%Such kind of ready simulation cares only stable consistent states.
%and it is fully abstract w.r.t conjunction and parallel composition on LLTS.

\begin{mydefn}[{[3]}]\label{D:RS}
Let $(P, Act_{\tau}, \rightarrow , F)$ be an  LLTS.
A relation ${\mathcal R} \subseteq P \times P$ is a stable ready simulation relation, if, for any $(p,q) \in {\mathcal R}$ and $a \in Act $,
\textbf{(RS1)} both $p$ and $q$ are stable;
\textbf{(RS2)} $p \notin F$ implies $q \notin F$;
\textbf{(RS3)} $p \stackrel{a}{\Rightarrow}_F|p'$ implies $\exists q'.q \stackrel{a}{\Rightarrow}_F|q'\; \textrm{and}\;(p',q') \in {\mathcal R}$;
\textbf{(RS4)} $p\notin F$ implies ${\mathcal I}(p)={\mathcal I}(q)$.

 We say that $p$ is stable ready simulated by $q$, in symbols $p \underset{\thicksim}{\sqsubset}_{RS} q$, if there exists a stable ready simulation relation $\mathcal R$ with $(p,q) \in {\mathcal R}$.
 Further, $p$ is ready simulated by $q$, written $p\sqsubseteq_{RS}q$, if
 $\forall p'(p\stackrel{\epsilon}{\Rightarrow}_F| p' \;\text{implies}\; \exists q'(q \stackrel{\epsilon}{\Rightarrow}_F| q'\; \text{and}\;p' \underset{\thicksim}{\sqsubset}_{RS} q'))$.
 The kernels of $\underset{\thicksim}{\sqsubset}_{RS}$ and $\sqsubseteq_{RS}$ are denoted by $\approx_{RS}$ and $=_{RS}$ resp..
% It is easy to see that $\underset{\thicksim}{\sqsubset}_{RS}$ itself is a stable ready simulation relation and both $\underset{\thicksim}{\sqsubset}_{RS}$ and $\sqsubseteq_{RS}$ are pre-order.
\end{mydefn}

%\subsection{The calculus $\text{CLL}_R$ and its operational semantics}

%This subsection recalls the LLTS-oriented process calculus $\text{CLL}_R$ presented in \cite{Zhang14}.
Next we fix some notations and terminologies related to $\text{CLL}_R$ and recall some results obtained in [5].
Let $V_{AR}$ be an infinite set of variables.
Terms of $\text{CLL}_{R}$ are given by the BNF grammar:
\[t::= 0 \mid \perp \mid (\alpha.t) \mid (t\Box t)\mid(t\wedge t)\mid(t\vee t)\mid(t\parallel_A t)\mid X\mid \langle Z|E \rangle,\]
where $X \in V_{AR}$, $\alpha\in Act_\tau$, $A\subseteq Act$ and recursive specification
$E = E(V)$ with $V \subseteq V_{AR}$ is a set of equations $\{Y = t: Y \in V\}$ and $Z$ is a variable in $V$ that acts as the initial variable.
We often denote $\langle X|\{X=t\}\rangle$ briefly by $\langle X|X=t\rangle$.
In addition to standard operators in CCS and CSP, operators $\perp$, $\wedge$ and $\vee$ are introduced in $\text{CLL}_R$:
%Most of these operators are from CCS and CSP:
%0 is the process capable of doing no action;
%$\alpha.\_$ is action prefixing;
%$\Box$ is non-deterministic external choice;
%$\parallel_A $ is a CSP-style parallel composition.
$\bot$ represents an inconsistent process;
$\vee$ and $\wedge$ are  used to describe logical combinations of processes.

For any term $\langle Z|E \rangle$ with $E=E(V)$, each variable in  $V$ is bound with scope $E$.
This induces the notion of free occurrence of variable, bound (and free) variables and $\alpha$-equivalence as usual.
%A term $t$ is a \emph{process} if it is closed, that is, it contains no free variable.
 %\footnote{It is defined as usual.}.
The set of all processes (i.e., closed terms) is denoted by $T(\Sigma_{\text{CLL}_R} )$.
We use $p,q,r$ to represent processes.
Throughout this note, we assume that recursive variables are distinct from each other and no recursive variable has free occurrence; moreover we don't distinguish between $\alpha$-equivalent terms and use $\equiv$ for both syntactical identical and $\alpha$-equivalence.
 For any $t$, the term
 $t\{\langle X|E \rangle/X: X \in V\}$ is denoted briefly by $\langle t|E \rangle$.
%For example, consider $t \equiv X \Box a.\langle Y | Y = X \ \Box Y \rangle$ and $E(\{X\})=\{X=t'\}$ then $\langle t| E\rangle \equiv \langle X|X =t'\rangle \Box a.\langle Y | Y = \langle X|X=t'\rangle \Box Y \rangle$.
%In particular, for any $E(V)$ and $t \equiv X$, $\langle t|E \rangle \equiv \langle X|E\rangle$ whenever $X \in V$ and $\langle t|E \rangle \equiv X$ if $X \notin V$.
A context $C_{\widetilde{X}}$ is a term whose free variables are in $n$-tuple distinct variables $\widetilde{X}=(X_1,...,X_n)$ with $n \geq 0$.
 Given $\widetilde{p}=(p_1,\dots,p_n)$, the term $C_{\widetilde{X}}\{\widetilde{p}/\widetilde{X}\}$ is obtained from $C_{\widetilde{X}}$ by replacing $X_i$ by $p_i$ for each $i \leq n$ simultaneously.
 %In particular, we use $C_{\widetilde{X}}\{p/\widetilde{X}\}$ to denote the result of replacing all variables in $\widetilde{X}$ by $p$.
 % A context $C_{\widetilde{X}}$ is stable if $C_{\widetilde{X}}\{\widetilde{0}/\widetilde{X}\} \not\stackrel{\tau}{\rightarrow} $.

Given a term $t$, a variable $X$ is strongly (or weakly) guarded in $t$ if each occurrence of $X$ is within some subexpression $a.t_1$ ($\tau.t_1$ or $t_1 \vee t_2$ resp.).
 As usual, we assume that all recursive specifications (say $E(V)$) considered in the sequel are guarded (that is, for each $X \in V$ and $Z = t_Z \in E(V)$, each occurrence of $X$ is within some subexpression $a.t_1$ or $\tau.t_1$ or $t_1 \vee t_2$).
% guarded in
% we and such occurrence is guarded if it is weakly or strongly guarded.
%Given $t$, a variable $X$ is strongly (or weakly) guarded in $t$ if each occurrence of $X$ is strongly (weakly resp.) guarded, and $X$ is guarded if each occurrence of $X$ is guarded.
%%Moreover, (strongly or weakly) guarded variable is defined as usual.
%A recursive specification $E(V)$ is guarded if  $t_Z$.

SOS rules of $\text{CLL}_R$ are divided into two parts: operational rules and predicate rules.
Here we only list these rules in Table~\ref{Ta:OPERATIONAL_RULES}.
For motivation behind these rules, we refer the reader to [5].
%, where $a \in Act$, $\alpha \in Act_{\tau}$ and $A \subseteq Act$,

%Operational rules specify behaviours of processes.
%Negative premises in Rules $Ra_2$, $Ra_3$, $Ra_{13}$ and $Ra_{14}$ give $\tau$-transition precedence over visible transitions, which guarantees that the transition model of $\text{CLL}_{R}$ is $\tau$-pure.
%% in usual process calculus.
%Rule $Ra_6$ reflects that conjunction operator is a synchronous product for visible transitions.
%Rules $Ra_9$ and $Ra_{10}$ illustrate that the operational aspect of $t_1\vee t_2$ is same as internal choice.
%The operational rules of the other operators are as usual.
%
%
%
%Predicate rules specify  the inconsistency predicate $F$.
%Rule $Rp_1$ says that $\bot$ is inconsistent.
%Hence $\bot$ cannot be implemented.
%While $0$ is consistent and implementable.
%Thus $0$ and $\bot$ represent different processes.
%Rule $Rp_3$ reflects that if both two disjunctive parts are inconsistent then so is the disjunction.
%Rules $Rp_4-Rp_9$ describe the system design strategy that if one part is inconsistent, then so is the whole composition.
%Rules $Rp_{10}$ and $Rp_{11}$ reveal that a stable conjunction is inconsistent whenever its conjuncts have distinct ready sets.
%Rules $Rp_{13}$ and $Rp_{15}$ %\footnote{Notice that the transition relation $\stackrel{\epsilon}{\Rightarrow}|$ occurring in these two rules does not involve any requirement on consistency, see footnote~1. }
%are used to capture (LTS2) in Def.~\ref{D:LLTS}.
%Intuitively, these two rules say that if all stable $\tau$-descendants of $z$  are inconsistent, then $z$ itself is inconsistent.

\begin{table}[rht]
%\begin{center}
\rule{\textwidth}{0.5pt}
\noindent \textbf{Operational rules}\\
    $\begin{array}{llll}
          Ra_1\frac{-}{\alpha.x_1 \stackrel{\alpha}{\rightarrow} x_1}  &
          Ra_2\frac{x_1 \stackrel{a}{\rightarrow} y_1, x_2 \not \stackrel{\tau}{\rightarrow}}{x_1 \Box x_2 \stackrel{a}{\rightarrow} y_1} &
        Ra_3\frac{x_1 \not\stackrel{\tau}{\rightarrow} , x_2 \stackrel{a}{\rightarrow} y_2 }{x_1 \Box x_2 \stackrel{a}{\rightarrow} y_2} &
         Ra_4\frac{x_1 \stackrel{\tau}{\rightarrow} y_1}{x_1 \Box x_2 \stackrel{\tau}{\rightarrow} y_1 \Box x_2} \\
          Ra_5\frac{x_2 \stackrel{\tau}{\rightarrow} y_2}{x_1 \Box x_2 \stackrel{\tau}{\rightarrow} x_1 \Box y_2}&
        Ra_6\frac{x_1 \stackrel{a}{\rightarrow} y_1, x_2 \stackrel{a}{\rightarrow}y_2}{x_1 \wedge x_2 \stackrel{a}{\rightarrow} y_1 \wedge y_2}&
          Ra_7\frac{x_1 \stackrel{\tau}{\rightarrow} y_1}{x_1 \wedge x_2 \stackrel{\tau}{\rightarrow} y_1 \wedge x_2} &
          Ra_8\frac{x_2 \stackrel{\tau}{\rightarrow} y_2}{x_1 \wedge x_2 \stackrel{\tau}{\rightarrow} x_1 \wedge y_2}\\
         Ra_9\frac{-}{x_1 \vee x_2 \stackrel{\tau}{\rightarrow} x_1}&
         Ra_{10}\frac{-}{x_1 \vee x_2 \stackrel{\tau}{\rightarrow} x_2} &
             Ra_{11}\frac{x_1 \stackrel{\tau}{\rightarrow} y_1}{x_1 \parallel_A x_2 \stackrel{\tau}{\rightarrow} y_1\parallel_A x_2} &
         Ra_{12}\frac{x_2 \stackrel{\tau}{\rightarrow} y_2}{x_1 \parallel_A x_2 \stackrel{\tau}{\rightarrow} x_1 \parallel_A y_2}
  \end{array}$
  $\begin{array}{ll}
          Ra_{13}\frac{x_1 \stackrel{a}{\rightarrow} y_1 , x_2 \not \stackrel{\tau}{\rightarrow} }{x_1 \parallel_A x_2 \stackrel{a}{\rightarrow} y_1 \parallel_A x_2}(a\notin A)&
       \qquad  Ra_{14}\frac{x_1 \not\stackrel{\tau}{\rightarrow} , x_2 \stackrel{a}{\rightarrow} y_2 }{x_1 \parallel_A x_2 \stackrel{a}{\rightarrow} x_1 \parallel_A y_2}(a\notin A)\\
          Ra_{15}\frac{x_1 \stackrel{a}{\rightarrow} y_1, x_2 \stackrel{a}{\rightarrow}y_2}{x_1\parallel_A x_2 \stackrel{a}{\rightarrow} y_1 \parallel_A y_2} (a\in A)&
        \qquad  Ra_{16}\frac{\langle t_X| E \rangle  \stackrel{\alpha}{\rightarrow} y}{\langle X|E \rangle \stackrel{\alpha}{\rightarrow} y}(X=t_X \in E)
     \end{array}$
\noindent $\;$\\

\noindent \textbf{Predicative rules} \\
$\begin{array}{lllll}
           Rp_1\frac{-}{\bot F}&
            \quad Rp_2\frac{x_1 F}{\alpha .x_1 F} &
           \quad Rp_3\frac{x_1 F, x_2 F}{x_1\vee x_2 F}&
           \quad  Rp_4\frac{x_1 F}{x_1\Box x_2 F}&
           \quad Rp_5\frac{x_2 F}{x_1\Box x_2 F}\\
            Rp_6\frac{x_1 F}{x_1\parallel_A x_2 F}&
         \quad  Rp_7\frac{x_2 F}{x_1\parallel_A x_2 F}&
          \quad  Rp_8\frac{x_1 F}{x_1\wedge x_2 F}&
          \quad  Rp_9\frac{x_2 F}{x_1\wedge x_2 F}
    \end{array}$
    $\begin{array}{lll}
           Rp_{10}\frac{x_1 \stackrel{a}{\rightarrow} y_1, x_2 \not\stackrel{a}{\rightarrow}, x_1 \wedge x_2 \not\stackrel{\tau}{\rightarrow}}{x_1 \wedge x_2 F}&
            Rp_{11}\frac{x_1 \not\stackrel{a}{\rightarrow} , x_2 \stackrel{a}{\rightarrow} y_2, x_1 \wedge x_2 \not\stackrel{\tau}{\rightarrow}}{x_1 \wedge x_2 F}&
           Rp_{12}\frac{x_1 \wedge x_2 \stackrel{\alpha}{\rightarrow} z, \{yF:x_1 \wedge x_2 \stackrel{\alpha}{\rightarrow}y\}}{x_1 \wedge x_2 F} \\
           Rp_{13}\frac{\{yF:x_1 \wedge x_2 \stackrel{\epsilon}{\Rightarrow}|y\}}{x_1 \wedge x_2 F}  &
            Rp_{14}\frac{\langle t_X|E \rangle F}{\langle X|E \rangle F}(X = t_X \in E) &
           Rp_{15}\frac{\{yF:\langle X|E \rangle \stackrel{\epsilon}{\Rightarrow}|y\}}{\langle X|E \rangle F}
    \end{array}
    $

\rule{\textwidth}{0.5pt}
\caption{SOS rules of $\text{CLL}_R$\label{Ta:OPERATIONAL_RULES}}
%\end{center}
\end{table}

 The calculus $\text{CLL}_R$ has the unique stable transition model (denoted by $M_{\text{CLL}_R}$), which exactly consists of all positive literals of the form $t \stackrel{\alpha}{\rightarrow}t'$ or $tF$ that are provable in $Strip(\text{CLL}_{R},M_{\text{CLL}_{R}})$ [5].
 Here $Strip(\text{CLL}_{R},M_{\text{CLL}_{R}})$ is the stripped version [1] of $\text{CLL}_{R}$ w.r.t $M_{\text{CLL}_{R}}$.
 Each rule in $Strip(\text{CLL}_{R},M_{\text{CLL}_{R}})$ is of the form $\frac{pprem(r)}{conc(r)}$ for some ground instance  $r$ of rules in Table~\ref{Ta:OPERATIONAL_RULES} such that $M_{\text{CLL}_{R}}\models nprem(r)$, where  $nprem(r)$ (or, $pprem(r)$) is the set of negative (positive resp.) premises of $r$, $conc(r)$ is the conclusion of $r$ and $M_{\text{CLL}_{R}}\models nprem(r)$ means that for each $t\not\stackrel{\alpha}{\rightarrow} \in nprem(r)$, $t\stackrel{\alpha}{\rightarrow}s \notin M_{\text{CLL}_{R}}$ for any $s$.
 The notion of proof tree in $Strip(\text{CLL}_{R},M_{\text{CLL}_{R}})$ is defined as usual [1].
 Notice that all proof trees are well-founded, and such fact will play central role in demonstrating the consistency of processes.
 Based on  $M_{\text{CLL}_{R}}$, we can get the LTS $(T(\Sigma_{\text{CLL}_{R}}),Act_{\tau},\rightarrow_{\text{CLL}_{R}},F_{\text{CLL}_{R}})$ ($LTS(\text{CLL}_{R})$ for short) in the standard way (e.g., [1]).
%    The LTS associated with $\text{CLL}_{R}$, in symbols $LTS(\text{CLL}_{R})$, is the quadruple
%    $(T(\Sigma_{\text{CLL}_{R}}),Act_{\tau},\rightarrow_{\text{CLL}_{R}},F_{\text{CLL}_{R}})$, where  $p\in F_{\text{CLL}_{R}}$ iff $pF \in M_{\text{CLL}_{R}}$, and $p \stackrel{\alpha}{\rightarrow}_{\text{CLL}_{R}} p'$ iff $p\stackrel{\alpha}{\rightarrow} p' \in M_{\text{CLL}_{R}}$.
%    Therefore $p \stackrel{\alpha}{\rightarrow}_{\text{CLL}_{R}} p'$ (or, $p \in F_{\text{CLL}_R}$) iff  $Strip( \text{CLL}_{R}, M_{\text{CLL}_{R}}) \vdash p\stackrel{\alpha}{\rightarrow}p'$ ($pF$ resp.) for any $p,p'$ and $\alpha \in Act_{\tau}$.
    %Clearly, every proof tree of $p\stackrel{\alpha}{\rightarrow}p'$ and $pF$ in $Strip( \text{CLL}_{R}, M_{\text{CLL}_{R}})$ is well-founded.
    For simplicity,  we always omit the subscripts in $\stackrel{\alpha}{\rightarrow}_{\text{CLL}_{R}}$ and $F_{\text{CLL}_{R}}$.
We end this section by recalling some fundamental properties of $LTS(\text{CLL}_{R})$, which are asserted by  Theorems~4.1 and 6.1 and Lemma~4.2 in [5].

%\begin{lemma}\label{L:F_NORMAL}
%Let $p$ and $q$ be any two processes. Then
%
%    \noindent (1) $p \vee q \in F $  iff $p,q \in F $;\\
%    \noindent (2)  $\alpha.p \in F $ iff $p \in F $ for each $\alpha \in Act_{\tau}$;\\
%    \noindent (3)  $p \odot q \in F $  iff either $p \in F $ or $q \in F $  with $\odot \in \{\Box, \parallel_A\}$;\\
%    \noindent (4) $p \in F $ or $q \in F $ implies $p \wedge q \in F $;\\
%    \noindent (5) $0 \notin F $ and $\bot \in F $;\\
%    \noindent (6) $\langle X|E \rangle \in F$ iff $\langle t_X|E\rangle \in F$ for each $X$ with $X=t_X \in E$.
%\end{lemma}

 \begin{theorem}\label{L:LLTS}
 (1) $LTS({\text{CLL}_{R}})$ is a $\tau$-pure LLTS.
 %(1) The LTS associated with $\text{CLL}_R$ is a $\tau$-pure LLTS.
\noindent (2) If $p\in F$ and $\tau\in \mathcal{I}(p)$ then $\forall q(p\stackrel{\tau}{\rightarrow}q \;\text{implies}\; q \in F)$, and hence $p \sactions{\epsilon}q$ and $q \notin F$ implies $p \fsactions{\epsilon}q$.
\noindent (3) If $p \sqsubseteq_{RS} q$ then $C_X\{p/X\} \sqsubseteq_{RS} C_X\{q/X\}$ for any $C_X$, and hence, if $p \sqsubseteq_{RS} q$ and  $C_X\{p/X\} \notin F$ then $C_X\{q/X\} \notin F$.
\end{theorem}
%\begin{proof}
%See Theorem~4.1, Lemma~4.2 and Theorem~6.1 in \cite{Zhang14}.
%\end{proof}

%\begin{theorem}\label{L:precongruence}
%
%\end{theorem}

\section{Main results}
In [5], the following theorem has been obtained.

\begin{theorem}[Unique solution]
%\noindent \textbf{Theorem} (Unique solution)
For any $p,q \notin F $ and $t_X$ where $X$ is strongly guarded and does not occur in the scope of any conjunction, if $p =_{RS} t_X\{p/X\}$ and $q =_{RS} t_X\{q/X\}$ then $p =_{RS} q$.
Moreover $\langle X | X=t_X \rangle$ is the unique consistent solution (modulo $=_{RS}$) of the equation $X =_{RS} t_X$ whenever  consistent solutions exist.
\end{theorem}

The next example demonstrates that this theorem no longer holds if we drop the assumption that $X$ does not occur in the scope of any conjunction.
%As we know, temporal operators could be described in equational style, represented by fixpoints of some equations \cite{ModelChecking}.
%Such style requires us to remove the special requirement (i.e. $X$ does not occur in the scope of any conjunction) occurring in Theorem Unique Solution.
%In the following, we give an example to show that this removement may result in multi consistent solutions.

\begin{example}\label{E:MULTI_SOLUTION}
  Consider the equation $X=_{RS}t_X$ where $t_X \equiv (\langle Y |Y =a.Y\rangle \wedge a.X) \vee (\langle Z |Z =b.Z\rangle \wedge b.X)$.
  % with $a,b\in Act$.
  Clearly, $X$ is strongly guarded in $t_X$.
  We shall show that both $\langle X|X=a.X \rangle$ and $\langle X|X=b.X \rangle$ are consistent solutions.% of this equation.

 Let us first prove that $\langle X|X=a.X \rangle \notin F$.
 On the contrary, suppose that $\langle X|X=a.X \rangle \in F$. Then the last rule applied in the proof tree of %$Strip(\text{CLL}_R,M_{\text{CLL}_R}) \vdash \langle X| X = a.X \rangle F$
 $\langle X| X = a.X \rangle F$ is either $\frac{a.\langle X| X = a.X  \rangle F}{\langle X| X =a.X \rangle F}$ or $\frac{\{rF:\langle X| X =a.X\rangle  \stackrel{\epsilon}{\Rightarrow}|r\}}{\langle X| X =a.X \rangle F}$.
  Then it is easy to see that every proof tree of $\langle X| X = a.X \rangle F$ has a proper subtree with root $\langle X| X = a.X \rangle F$, this contradicts the well-foundedness of proof tree, as desired.

  Secondly we show that $\langle X|X=a.X \rangle$ is a solution.
  Analysis similar to that above shows that $\langle Y|Y=a.Y \rangle \wedge a.\langle X|X=a.X\rangle \notin F$ and $\langle Y|Y=a.Y \rangle \wedge \langle X|X=a.X\rangle \notin F$.
  Then it is easy to check that the binary relation $\mathcal R$ given below is a stable ready simulation relation, where $P_v \triangleq \langle v | v = a.v \rangle$ with $v \in \{X,Y \}$.
%  \begin{multline*}
%    \mathcal{R}\triangleq\{(\langle X|X=a.X \rangle,\langle Y|Y=a.Y \rangle \wedge a.\langle X|X=a.X\rangle),\\(\langle X|X=a.X \rangle,\langle Y|Y=a.Y \rangle \wedge \langle X|X=a.X\rangle), \\
%    (\langle Y|Y=a.Y \rangle \wedge a.\langle X|X=a.X\rangle,\langle X|X=a.X \rangle),\\
%    (\langle Y|Y=a.Y \rangle \wedge \langle X|X=a.X\rangle,\langle X|X=a.X \rangle)\}.
%  \end{multline*}
 \[
    \mathcal{R}\triangleq\{(P_X,P_Y \wedge a.P_X),(P_X,P_Y \wedge P_X),
    (P_Y \wedge a.P_X,P_X),
    (P_Y   \wedge P_X,P_X)\}.\]
 Hence
 \[\langle X|X=a.X \rangle =_{RS}\langle Y|Y=a.Y \rangle \wedge a.\langle X|X=a.X\rangle.  \tag{\ref{E:MULTI_SOLUTION}.1}\]
  Moreover, $\langle Z |Z =b.Z\rangle \wedge b.\langle X|X=a.X \rangle \in F$ by Rules $Rp_{10}$, $Rp_{11}$ and $Rp_{12}$, which, together with (\ref{E:MULTI_SOLUTION}.1), implies  $\langle X|X=a.X\rangle =_{RS}t_X\{\langle X|X=a.X \rangle /X\}$.

  Summarily, $\langle X|X=a.X \rangle$ is a consistent solution.
  Similarly,  so is $\langle X|X=b.X \rangle$.
  However, $\langle X|X=a.X \rangle \not=_{RS}\langle X|X=b.X \rangle$. \qed
\end{example}

 % Consequently, unlike usual process calculus, it does not always hold that there exists a unique solution of a given equation $X=t$ in $\text{CLL}_R$ even if the variable $X$ is strongly guarded in $t$.

 For any equation $X=_{RS}t_X$, it is obvious that $\langle X|X=t_X \rangle$ is a solution of this equation.
 Moreover, the preceding example reveals that there may be  more than one (consistent) solution.
 Then it is natural to try to relate $\langle X|X=t_X \rangle$ to other solutions.
 As the main result of this note, we intend to show that, if $X$ is strongly guarded in $t_X$ then $\langle X|X=t_X \rangle$ is the greatest  solution of the equation $X=_{RS}t_X$.
  %and hence $\langle X|X=t_X \rangle$ is consistent if and only if consistent solution exits.
 In other words, $\langle X |X=t_X \rangle$ captures the loosest solution whenever $X$ is strongly guarded in $t_X$.
  To this end, a few of results in [5] are recalled below.
  %, which will be useful in the sequel.
  The following facts are confirmed by Lemmas~5.6-5.8 in [5].

  \begin{lemma}\label{L:ONE_ACTION_TAU}
   If $C_X\{p/X\} \stackrel{\alpha}{\rightarrow} r$ then

     (1) if $\alpha = \tau$ then
     either (1.1) there exists $C_X'$ such that $r\equiv C_X'\{p/X\}$ and $C_X\{q/X\} \stackrel{\tau}{\rightarrow}C_X'\{q/X\}$ for any $q$,
     or (1.2) there exist $C_{X,Z}'$ and $p'$ such that $p \action{\tau} p'$, $r \equiv C_{X,Z}'\{p/X,p'/Z\}$ and $C_X\{q/X\} \action{\tau} C_{X,Z}'\{q/X,q'/Z\}$ for any $q \action{\tau} q'$;

    (2) if $\alpha \in Act$ then there exits $C_{X,\widetilde{Y}}'$ such that
    (2.1) $r \equiv C_{X,\widetilde{Y}}'\{p/X,\widetilde{p_Y'}/\widetilde{Y}\}$ for some $\widetilde{p_Y'}$ with $p \action{\alpha} p_Y'$ for each $Y \in \widetilde{Y}$;
    (2.2) if $C_X\{q/X\}$ is stable and $q \action{\alpha} q_Y'$  for each $Y \in \widetilde{Y}$, then $C_X\{q/X\}\action{\alpha} C_{X,\widetilde{Y}}'\{q/X,\widetilde{q_Y'}/\widetilde{Y}\}$;

   (3) in particular, if $X$ is guarded in $C_X$ then there exists $B_X$ such that $r \equiv B_X\{p/X\}$ and  for any $q$, $C_X\{q/X\} \action{\alpha}B_X\{q/X\}$. \qed
  \end{lemma}
%  \begin{proof}
%See Lemmas~5.6, 5.7 and 5.8 in \cite{Zhang14}.
%  \end{proof}

%  \begin{lemma}\label{L:ONE_ACTION_VISIBLE}
%    Let $a \in Act$. If $C_X\{p/X\} \action{a} r$ then there exits $C_{X,\widetilde{Y}}'$ such that
%
%    \noindent (1) $r \equiv C_{X,\widetilde{Y}}'\{p/X,\widetilde{p_Y'}/\widetilde{Y}\}$ for some $\widetilde{p_Y'}$ with $p \action{a} p_Y'$ for each $Y \in \widetilde{Y}$, and
%
%    \noindent (2) if $C_X\{q/X\}$ is stable and $q \action{a} q_Y'$  for each $Y \in \widetilde{Y}$, then $C_X\{q/X\}\action{a} C_{X,\widetilde{Y}}'\{q/X,\widetilde{q_Y'}/\widetilde{Y}\}$.
%  \end{lemma}
% \begin{proof}
%See Lemma~5.7 in \cite{Zhang14}.
%  \end{proof}

%  \begin{lemma}\label{L:ONE_ACTION_VISIBLE_GUARDED}
%    Let $X$ be guarded in $C_X$.
%    If $C_X\{p/X\} \action{\alpha} r$ then there exists $B_X$ such that $r \equiv B_X\{p/X\}$ and  for any $q$, $C_X\{q/X\} \action{\alpha}B_X\{q/X\}$.
%  \end{lemma}
% \begin{proof}
%See Lemma~5.8 in \cite{Zhang14}.
%  \end{proof}
The next property is asserted by Lemmas~5.6, 5.8 and 5.14 in [5].
  \begin{lemma}\label{L:MULTI_TAU_GF_STABLE}
    If $C_X\{p/X\}\sactions{\epsilon}r$ then there exist $C_{X,\widetilde{Y}}'$ and $p_Y'$ for $Y \in \widetilde{Y}$ such that
     %\footnote{If $C_X\{p/X\}\not\stackrel{\tau}{\rightarrow}$ then  it is trivial that $\widetilde{Y}= \emptyset$.}
    \noindent (1) $p \sactions{\tau}p_Y'$ for each $Y \in \widetilde{Y}$ and $r\equiv C_{X,\widetilde{Y}}'\{p/X,\widetilde{p_Y'}/\widetilde{Y}\}$;
    \noindent (2) for any $q$ such that $q \action{\tau}$ iff $p\action{\tau}$, if $q \sactions{\tau}q_Y'$ for each $Y \in \widetilde{Y}$ then $C_X\{q/X\}\sactions{\epsilon}C_{X,\widetilde{Y}}'\{q/X,\widetilde{q_Y'}/\widetilde{Y}\}$;
    \noindent (3) in particular, if $X$ is strongly guarded in $C_X$ then so it is in $C_{X,\widetilde{Y}}'$, $\widetilde{Y}= \emptyset$ and $C_X\{q/X\}\sactions{\epsilon}C_{X,\widetilde{Y}}'\{q/X\}$ for any $q$. \qed
  \end{lemma}
% \begin{proof}
% Immediately follows from Lemmas~5.6, 5.8 and 5.14 in \cite{Zhang14}.
%  \end{proof}

%Notice that, in Lemma~\ref{L:MULTI_TAU_GF_STABLE}(2), it is necessary to make the assumption that $q \action{\tau}$ iff $p\action{\tau}$. For example, let $C_X \equiv \langle Y|Y =X \Box a.Y \rangle$, $p \equiv 0$ and $q \equiv \tau.0$. Then $C_X\{p/X\} \stackrel{\epsilon}{\Rightarrow}|C_X\{p/X\} $ and $C_X\{q/X\} \stackrel{\epsilon}{\Rightarrow}|0 \Box a.C_X\{q/X\} $ but it does not hold that $C_X\{q/X\} \stackrel{\epsilon}{\Rightarrow}|C_X\{q/X\} $.

%  As mentioned in Section~2, a distinguishing feature of LLTS is that it involves consideration of inconsistency.
 % To prove the main result of this note, the next lemma  is needed.

  \begin{lemma}\label{L:precongruence_F}
    If $X$ is strongly guarded in $t_X$ and $p\sqsubseteq_{RS}t_X\{p/X\}$ then for any context $C_Y$, $C_Y\{t_X\{p/X\}/Y\} \notin F$ implies $C_Y\{\langle X|X=t_X \rangle/Y\} \notin F$.
  \end{lemma}
  \begin{proof}
 % Since $X$ is guarded in $t_X$,
 By Lemma~\ref{L:ONE_ACTION_TAU}(3) and $Ra_{16}$, we have
$\mathcal{I}(\langle X|X=t_X \rangle) = \mathcal{I}(t_X\{p/X\})$.
  Then, by Lemma~\ref{L:ONE_ACTION_TAU}(1)(2), for any context $D_Y^*$, we get
%\begin{multline*}
%\text{for any }D_Y \\
\[ \mathcal{I}(D_Y^*\{t_X\{p/X\}/Y\})= \mathcal{I}(D_Y^*\{\langle X|X=t_X \rangle /Y\}). \tag{\ref{L:precongruence_F}.1}\]
%\end{multline*}
    Set $\Omega \triangleq \{B_Y\{\langle X|X=t_X \rangle/Y\}:B_Y\text{ is a context and }
      B_Y\{t_X\{p/X\}/Y\} \notin F\}$.
   To complete the proof, it suffices to prove that $F \cap \Omega = \emptyset$.
   % Conversely, suppose that $F \cap \Omega \neq \emptyset$.
   % Due to the well-foundedness of proof trees, to complete the proof,
   We intend to show that, for each $t \in \Omega$, any proof tree of $tF$ %$Strip(\text{CLL}_R,M_{\text{CLL}_R})\vdash tF$
    % $C_Y\{\langle X|X=t_X \rangle /Y\} \in \Omega$, any proof tree for $Strip(\text{CLL}_R,M_{\text{CLL}_R})\vdash C_Y\{\langle X|X=t_X \rangle /Y\}F$
     has a proper subtree with root $sF$ for some $s \in \Omega$.
     Such statement implies $F \cap \Omega = \emptyset$. Otherwise, a contradiction arises due to the fact that proof trees are well-founded.
    Let $\mathcal T$ be any proof tree of $C_Y\{\langle X|X=t_X \rangle /Y\}F$ with $C_Y\{\langle X|X=t_X \rangle /Y\} \in \Omega$.
    Then
    \[C_Y\{t_X\{p/X\}/Y\}\notin F.\tag{\ref{L:precongruence_F}.2}\]
    The rest of the proof runs by distinguishing cases  based on $C_Y$.
    Here we handle only three non-trivial cases; the others %(i.e., $C_Y \equiv \alpha.B_Y, B_Y \odot D_Y$ with $\odot \in \{\Box,\parallel_A,\vee\}$)
    are left to the reader.

    \noindent \textbf{Case 1.} $C_Y \equiv Y$.
    %Then $C_Y\{\langle X|X=t_X \rangle /Y\} \equiv \langle X|X=t_X \rangle$. So
    Clearly, the last rule applied in $\mathcal T$ is either
 $\frac{\langle t_X|X=t_X \rangle F}{\langle X|X=t_X \rangle F}$ or  $\frac{\{rF:\langle X|X=t_X \rangle \sactions{\epsilon}r\}}{\langle X|X=t_X \rangle F}$.
  For the former, since $p\sqsubseteq_{RS}t_X\{p/X\}$, by (\ref{L:precongruence_F}.2) and Theorem~\ref{L:LLTS}(3), $t_X\{t_X\{p/X\}/X\} \notin F$.
    Hence $\mathcal T$ has a proper subtree with root $\langle t_X|X=t_X\rangle F$ and $\langle t_X|X=t_X\rangle \equiv t_X\{\langle X|X=t_X \rangle/X\} \in \Omega$, as desired.

   For the latter, we treat the non-trivial case where $\langle X|X=t_X \rangle \action{\tau}$.
    Since $t_X\{p/X\} \notin F$, by Theorem~\ref{L:LLTS}(1), $t_X\{p/X\}\fsactions{\epsilon}s$ for some $s$.
    For this transition,  by Lemma~\ref{L:MULTI_TAU_GF_STABLE}(3), there exists $t_X'$ such that $t_X\{\langle X|X=t_X\rangle/X\} \actions{\epsilon}| t_X'\{\langle X|X=t_X \rangle /X\}$ and $s \equiv t_X'\{p/X\}$.
    %Further, by Lemma~\ref{L:ONE_ACTION_VISIBLE_GUARDED},
   Then, by $Ra_{16}$ and $\langle X|X=t_X\rangle \action{\tau}$, we get
   $\langle X|X=t_X\rangle \sactions{\tau} t_X'\{\langle X|X=t_X \rangle/X\}$.
   So $\mathcal T$ has a proper subtree with root $t_X'\{\langle X|X=t_X \rangle/X\} F$.
    %and $X$ is strongly guarded in stable $t_X'$.
    Moreover, by Theorem~\ref{L:LLTS}(3), we have $t_X'\{t_X\{p/X\}/X\} \notin F$ because of $s\equiv t_X'\{p/X\} \notin F$ and $p \RS t_X\{p/X\}$.
    Hence $t_X'\{\langle X|X=t_X \rangle /X\} \in \Omega$, as desired.

    %\noindent Case 2. $C_Y \equiv \langle Z|E \rangle$.
%   Then the last rule applied in $\mathcal T$ is one of the following:
%
%   \noindent  Case 2.1. $\frac{\langle t_Z|E \rangle \{\langle X|X=t_X \rangle /Y\} F}{\langle Z|E \rangle \{\langle X|X=t_X \rangle /Y\} F}$.
%
%    By $Rp_{14}$, we get $\langle t_Z|E \rangle \{t_X\{p/X\}/Y\} \notin F $ due to (\ref{L:precongruence_F}.2).
%    So $\langle t_Z|E \rangle \{\langle X|X=t_X \rangle /Y\} \in \Omega$.
%
%    \noindent Case 2.2 $\frac{\{rF:\langle Z|E\rangle \{\langle X|X=t_X \rangle /Y\} \sactions{\epsilon}r\}}{\langle Z|E \rangle \{\langle X|X=t_X \rangle /Y\} F}$.

\noindent \textbf{Case 2.} $C_Y \equiv \langle Z|E \rangle$.
   Then the last rule applied in $\mathcal T$ is $\frac{\langle t_Z|E \rangle \{\langle X|X=t_X \rangle /Y\} F}{\langle Z|E \rangle \{\langle X|X=t_X \rangle /Y\} F}$ or $\frac{\{rF:\langle Z|E\rangle \{\langle X|X=t_X \rangle /Y\} \sactions{\epsilon}r\}}{\langle Z|E \rangle \{\langle X|X=t_X \rangle /Y\} F}$.
    For the former, we get $\langle t_Z|E \rangle \{t_X\{p/X\}/Y\} \notin F $ due to  $Rp_{14}$ and (\ref{L:precongruence_F}.2).
    So $\langle t_Z|E \rangle \{\langle X|X=t_X \rangle /Y\} \in \Omega$, as desired.

   For the latter, by (\ref{L:precongruence_F}.2) and Theorem~\ref{L:LLTS}(1), $C_Y\{t_X\{p/X\}/Y\}  \fsactions{\epsilon}s$ for some $s$.
    For this transition, there exist $C_{Y,\widetilde{W}}'$ and $\widetilde{s_W'}$ that satisfy clauses (1,2,3) in Lemma~\ref{L:MULTI_TAU_GF_STABLE}.
    Hence $s\equiv C_{Y,\widetilde{W}}'\{t_X\{p/X\}/Y,\widetilde{s_W'}/\widetilde{W}\}$ and for each $W\in \widetilde{W}$, $t_X\{p/X\} \sactions{\tau} s_W'$.
    For each such transition, say $t_X\{p/X\} \sactions{\tau}s_W'$, by Lemma~\ref{L:MULTI_TAU_GF_STABLE}(3) and \ref{L:ONE_ACTION_TAU}(3), $t_X\{\langle X|X=t_X \rangle /X\} \actions{\tau} |t_X'^W\{\langle X|X=t_X\rangle/X\}$ and $s_W' \equiv t_X'^W\{p/X\}$ for some $t_X'^W$.
    So, $\langle X|X=t_X \rangle   \sactions{\tau} t_X'^W\{\langle X|X=t_X\rangle/X\}$ for each $W \in \widetilde{W}$.
    Further, since $C_{Y,\widetilde{W}}'$  satisfies clause (2) in Lemma~\ref{L:MULTI_TAU_GF_STABLE}, by (\ref{L:precongruence_F}.1) with $D_Y^* \equiv Y$, we get $C_Y\{\langle X|X=t_X\rangle/Y\}  \sactions{\epsilon} u$, where
    $u \stackrel{\vartriangle}{\equiv}  C_{Y,\widetilde{W}}'\{\langle X|X=t_X \rangle/Y,\widetilde{t_X'^W}\{\langle X|X=t_X \rangle /X\}/\widetilde{W}\}$.
  Hence $\mathcal T$ has a proper subtree with root $uF$.
    Moreover, since $s\equiv C_{Y,\widetilde{W}}'\{t_X\{p/X\}/Y,\widetilde{t_X'^W\{p/X\}}/\widetilde{W}\} \notin F$ and $p \RS t_X\{p/X\}$, we obtain $C_{Y,\widetilde{W}}'\{t_X\{p/X\}/Y,\widetilde{t_X'^W}\{t_X\{p/X\}/X\}/\widetilde{W}\} \notin F$ due to   Theorem~\ref{L:LLTS}(3). Then $u \in \Omega$, as desired.

 \noindent   \textbf{Case 3.} $C_Y \equiv B_Y \wedge D_Y$.
 We distinguish four cases based on the last rule applied in $\mathcal T$.
 Since rules for $\wedge$ are symmetric w.r.t its operands, we consider only one of two symmetric rules.

 \noindent \textbf{Case 3.1.} $\frac{B_Y\{\langle X|X=t_X \rangle/Y\}F}{C_Y\{\langle X|X=t_X \rangle/Y\}F}$.
  %or $\frac{D_Y\{\langle X|X=t_X \rangle/Y\}F}{C_Y\{\langle X|X=t_X \rangle/Y\}F}$.
By (\ref{L:precongruence_F}.2) and  $Rp_{8}$, $B_Y\{t_X\{p/X\}/Y\} \notin F$ and hence $B_Y\{\langle X|X=t_X \rangle/Y\} \in \Omega$.
% Similar arguments apply to the later.
 %Moreover the premise of the root of $\mathcal T$ is $B_Y\{\langle X|X=t_X \rangle/Y\}F$,  as desired.\\

 \noindent \textbf{Case 3.2.} $\frac{B_Y\{\langle X|X=t_X \rangle/Y\} \action{a}r}{C_Y\{\langle X|X=t_X \rangle/Y\}F}$ with $D_Y\{\langle X|X=t_X \rangle/Y\} \not\action{a}$ and $C_Y\{\langle X|X=t_X \rangle/Y\}\not\action{\tau}$.
 Then, by (\ref{L:precongruence_F}.1), we get $B_Y\{t_X\{p/X\}/Y\} \action{a}$, $D_Y\{t_X\{p/X\}/Y\} \not\action{a}$ and $C_Y\{t_X\{p/X\}/Y\} \not\action{\tau}$.
  Thus $C_Y\{t_X\{p/X\}/Y\} \in F$ follows by $Rp_{10}$ and $Rp_{11}$, which contradicts (\ref{L:precongruence_F}.2). Hence this case is impossible.

 \noindent \textbf{Case 3.3.} $\frac{\{rF:C_Y\{\langle X|X=t_X \rangle/Y\} \sactions{\epsilon}r\}}{C_Y\{\langle X|X=t_X \rangle/Y\}F}$.
 Similar to the second alternative in the proof of Case~2, omitted.

 \noindent \textbf{Case 3.4.} $\frac{C_Y\{\langle X|X=t_X \rangle/Y\} \action{\alpha}r', \{rF:C_Y\{\langle X|X=t_X \rangle/Y\} \action{\alpha}r\}}{C_Y\{\langle X|X=t_X \rangle/Y\}F}$.
% So $\alpha \in \mathcal{I}(C_Y\{\langle X|X=t_X \rangle/Y\} )$.
 Then, by (\ref{L:precongruence_F}.1) with $D_Y^* \equiv C_Y$, (\ref{L:precongruence_F}.2) and Theorem~\ref{L:LLTS}(1), there exists $s$ such that
 \[C_Y\{t_X\{p/X\}/Y\} \faction{\alpha}s. \tag{\ref{L:precongruence_F}.3}\]
 In the following, we consider two cases based on $\alpha$.

 \noindent \textbf{Case~3.4.1.} $\alpha=\tau$.
    For the transition in (\ref{L:precongruence_F}.3), either (1.1) or (1.2) in Lemma~\ref{L:ONE_ACTION_TAU} holds.
    For the former, there exists $C_Y'$ such that $C_Y\{\langle X|X=t_X\rangle/Y\} \action{\tau} C_Y'\{\langle X|X=t_X\rangle/Y\}$ and $s\equiv C_Y'\{t_X\{p/X\}/Y\}$.
    Then $\mathcal T$ has a proper subtree with root $C_Y'\{\langle X|X=t_X\rangle/Y\}F$ and $C_Y'\{\langle X|X=t_X\rangle/Y\} \in \Omega$ due to $s \equiv C_Y'\{t_X\{p/X\}/Y\} \notin F$.

    Next we handle the latter where (1.2) in Lemma~\ref{L:ONE_ACTION_TAU} holds.
    In such situation, $s\equiv C_{Y,Z}'\{t_X\{p/X\}/Y,s'/Z\}$ for some $s',C_{Y,Z}'$ such that $t_X\{p/X\} \action{\tau} s'$ and
    \[C_Y\{q/Y\} \action{\tau} C_{Y,Z}'\{q/Y,q'/Z\} \text{ for any }q \action{\tau}q'.\tag{\ref{L:precongruence_F}.4}\]
    For $t_X\{p/X\} \action{\tau} s'$, by Lemma~\ref{L:ONE_ACTION_TAU}(3), there exists $t_X'$ such that $s'\equiv t_X'\{p/X\}$ and $t_X\{\langle X|X=t_X\rangle/X\} \action{\tau} t_X'\{\langle X|X=t_X\rangle/X\}$.
    Then by $Ra_{16}$, $\langle X|X=t_X\rangle \action{\tau} t_X'\{\langle X|X=t_X\rangle/X\}$.
    Further, it follows from (\ref{L:precongruence_F}.4) that $C_Y\{\langle X|X=t_X \rangle /Y\}\action{\tau} u$, where $u \stackrel{\vartriangle}{\equiv} C_{Y,Z}'\{\langle X|X=t_X \rangle /Y,t_X'\{\langle X|X=t_X\rangle/X\}/Z\}$.
    Thus $\mathcal T$ has a proper subtree with root $uF$.
    %Moreover, since $p \RS t_X\{p/X\}$ and $s\equiv C_{Y,Z}'\{t_X\{p/X\}/Y,t_X'\{p/X\}/Z\} \notin F$, by Theorem~\ref{L:LLTS}(3), $C_{Y,Z}'\{t_X\{p/X\}/Y,t_X'\{t_X\{p/X\}/X\}/Z\} \notin F$, which implies  $u\in \Omega$.
    Moreover, by Theorem~\ref{L:LLTS}(3) and $ s \notin F$, we get $C_{Y,Z}'\{t_X\{p/X\}/Y,t_X'\{t_X\{p/X\}/X\}/Z\} \notin F$, which implies  $u\in \Omega$.

 \noindent \textbf{Case~3.4.2.} $\alpha \in Act$.
 For the transition in (\ref{L:precongruence_F}.3),  there exists $C_{Y,\widetilde{Z}}'$ that satisfies (2.1) and (2.2) in Lemma~\ref{L:ONE_ACTION_TAU}(2).
 Thus $s \equiv C_{Y,\widetilde{Z}}'\{t_X\{p/X\}/Y,\widetilde{s_Z'}/\widetilde{Z}\}$ for some  $\widetilde{s_Z'}$ with $t_X\{p/X\} \action{\alpha} s_Z'$ for any $Z \in \widetilde{Z}$.
 For each such transition, say $t_X\{p/X\} \action{\alpha} s_Z'$, by Lemma~\ref{L:ONE_ACTION_TAU}(3), $s_Z' \equiv t_X'^Z\{p/X\}$ and $t_X\{\langle X|X=t_X\rangle/X\} \action{\alpha} t_X'^Z\{\langle X|X=t_X \rangle/X\}$ for some $t_X'^Z$.
 Then, by $Ra_{16}$, $\langle X|X=t_X\rangle \action{\alpha} t_X'^Z\{\langle X|X=t_X \rangle/X\}$ for $Z \in \widetilde{Z}$.
 %Moreover, by (\ref{L:precongruence_F}.1) with $D_Y^* \equiv C_Y$, $C_Y\{\langle X|X=t_X\rangle/Y\}$ is stable because of  (\ref{L:precongruence_F}.3) and $\alpha \in Act$.
 Further, since $C_{Y,\widetilde{Z}}'$ satisfies (2.2) in Lemma~\ref{L:ONE_ACTION_TAU}, by (\ref{L:precongruence_F}.1) with $D_Y^* \equiv C_Y$, we get
   $ C_Y\{\langle X|X=t_X \rangle/Y\} \action{\alpha}u$, where
   $u\stackrel{\vartriangle}{\equiv} C_{Y,\widetilde{Z}}'\{\langle X|X=t_X \rangle /Y, \widetilde{t_X'^Z}\{\langle X|X=t_X \rangle /X \}/\widetilde{Z}\}.$
 Thus $\mathcal T$ has a proper subtree with root $uF$.
 %Moreover, since $s \equiv C_{Y,\widetilde{Z}}'\{t_X\{p/X\}/Y,\widetilde{t_X'^Z\{p/X\}}/\widetilde{Z}\} \notin F$ and $p \RS t_X\{p/X\}$, we get $C_{Y,\widetilde{Z}}'\{t_X\{p/X\}/Y,\widetilde{t_X'^Z}\{t_X\{p/X\}/X\}/\widetilde{Z}\} \notin F$ by Theorem~\ref{L:LLTS}(3), and hence $u \in \Omega$.
 Moreover, by $s \notin F$ and Theorem~\ref{L:LLTS}(3), we get $C_{Y,\widetilde{Z}}'\{t_X\{p/X\}/Y,\widetilde{t_X'^Z}\{t_X\{p/X\}/X\}/\widetilde{Z}\} \notin F$, and hence $u \in \Omega$.
  \end{proof}

Having disposed of this preliminary step, we can now give a crucial result. Let us first recall a notion of up-to $\underset{\thicksim}{\sqsubset}_{RS}$, which depends on an equivalent formulation of $\RS$ provided by van Glabbeek [3].
 %Such formulation allows one to demonstrate that one process is ready simulated by another in coinductive-proof style.

%  \begin{mydefn}[\cite{Luttgen10}]\label{D:ALT_RS}
%A relation ${\mathcal R} \subseteq T(\Sigma_{\text{CLL}_R})\times T(\Sigma_{\text{CLL}_R})$ is an alternative ready simulation relation, if, for any $(p,q) \in {\mathcal R}$ and $a \in Act $,
%\textbf{(RSi)} $p \stackrel{\epsilon}{\Rightarrow}_F|p'$ implies $\exists q'.q \stackrel{\epsilon}{\Rightarrow}_F|q'\; \textrm{and}\;(p',q') \in {\mathcal R}$;
%\textbf{(RSiii)} $p \stackrel{a}{\Rightarrow}_F|p'$ and $p,q$ stable implies $\exists q'.q \stackrel{a}{\Rightarrow}_F|q'$ and $(p',q') \in {\mathcal R}$;
%\textbf{(RSiv)} $p\notin F $ and $p,q$ stable implies ${\mathcal I}(p)={\mathcal I}(q)$.
%
%We write $p \sqsubseteq_{ALT} q$ if there exists an alternative ready simulation relation $\mathcal R$ with $(p,q) \in \mathcal R$. It has been shown that  $\RS = \sqsubseteq_{ALT}$ \cite{Luttgen10}.
%\end{mydefn}
\begin{mydefn}[{[5]}]\label{D:ALT_RS_UPTO}
A relation ${\mathcal R} \subseteq T(\Sigma_{\text{CLL}_R})\times T(\Sigma_{\text{CLL}_R})$ is a ready simulation relation up to $\underset{\thicksim}{\sqsubset}_{RS}$ whenever, for any $(p,q) \in {\mathcal R}$ and $a \in Act $,

\noindent \textbf{(Upto-1)} $p \fsactions{\epsilon}p'$ implies $\exists q'.q \fsactions{\epsilon}q'$ and $p' \SRS {\mathcal R}\SRS q'$;

\noindent \textbf{(Upto-2)} $p \fsactions{a}p'$ and $p,q$ stable implies $\exists q'.q \fsactions{a}q'$ and $p'\SRS{\mathcal R}\SRS q'$;

\noindent \textbf{(Upto-3)} $p\notin F $ and $p,q$ stable implies ${\mathcal I}(p)={\mathcal I}(q)$.
\end{mydefn}

 This notion provides a sound up-to technique, that is, if $\mathcal R$ is a ready simulation relation up to $\SRS$, then ${\mathcal R} \subseteq \RS$ [5].
The next lemma asserts that $\langle X|X=t_X\rangle$ is the largest solution of the inequation $X\RS t_X$.
\begin{lemma}\label{L:GREATEST_SOLUTION}
 If $p \RS t_X\{p/X\}$ then $p \RS \langle X|X=t_X \rangle$ whenever $X$ is strongly guarded in $t_X$.
\end{lemma}
\begin{proof}
 By Lemma~\ref{L:ONE_ACTION_TAU}(3) and $Ra_{16}$, we get $\mathcal{I}(t_X\{p/X\})=\mathcal{I}(\langle X|X=t_X \rangle)$.
  Then, by  Lemma~\ref{L:ONE_ACTION_TAU}(1)(2),  it follows that, for any $D_Y$,
 \[ \mathcal{I}(D_Y\{t_X\{p/X\}/Y\})=\mathcal{I}(D_Y\{\langle X|X=t_X \rangle /Y\}).  \tag{\ref{L:GREATEST_SOLUTION}.1} \]
To complete the proof, it suffices to show that $t_X\{p/X\} \RS \langle X|X=t_X \rangle$. Set
$
   {\mathcal R} \triangleq \{(B_Y\{t_X\{p/X\}/Y\},B_Y\{\langle X|X=t_X \rangle/Y\}):
     B_Y\text{ is a context}\}$.
 We intend to prove that $\mathcal R$ is a ready simulation relation up to $\SRS$.

  Let $(C_Y\{t_X\{p/X\}/Y\},C_Y\{\langle X|X=t_X \rangle/Y\}) \in \mathcal R$.
   We shall check that such pair satisfies (Upto-1,2,3).
   For (Upto-3), it is obvious due to (\ref{L:GREATEST_SOLUTION}.1).

  \textbf{(Upto-1)} Assume $C_Y\{t_X\{p/X\}/Y\} \fsactions{\epsilon} s$.
   So there exist $C_{Y,\widetilde{Z}}'$ and $\widetilde{s_Z'}$ that satisfy clauses (1)-(3) in Lemma~\ref{L:MULTI_TAU_GF_STABLE}.
    Then $t_X\{p/X\} \sactions{\tau} s_Z'$ for $Z \in \widetilde{Z}$ and $s \equiv C_{Y,\widetilde{Z}}'\{t_X\{p/X\}/Y,\widetilde{s_Z'}/\widetilde{Z}\}$.
 For each such transition, say $t_X\{p/X\} \sactions{\tau} s_Z'$, by Lemma~\ref{L:MULTI_TAU_GF_STABLE}(3) and \ref{L:ONE_ACTION_TAU}(3), there exists $t_X'^Z$  with strongly guarded $X$ such that $s_Z' \equiv t_X'^Z\{p/X\}$ and $t_X\{\langle X|X=t_X\rangle /X\} \sactions{\tau} t_X'^Z\{\langle X|X=t_X \rangle /X\}$.
  So, by $Ra_{16}$, $\langle X|X=t_X\rangle \sactions{\tau} t_X'^Z\{\langle X|X=t_X \rangle /X\}$ for $Z \in \widetilde{Z}$.
  Since $C_{Y,\widetilde{Z}}'$ satisfies clause (2) in Lemma~\ref{L:MULTI_TAU_GF_STABLE}, by (\ref{L:GREATEST_SOLUTION}.1) with $D_Y \equiv Y$, we get
  $C_Y\{\langle X|X=t_X \rangle /Y\} \sactions{\epsilon}u$,
   where  $u \stackrel{\vartriangle}{\equiv} C_{Y,\widetilde{Z}}'\{\langle X|X=t_X \rangle /Y, \widetilde{t_X'^Z}\{\langle X|X=t_X \rangle /X\} /\widetilde{Z}\}$.
Put $w \stackrel{\vartriangle}{\equiv} C_{Y,\widetilde{Z}}'\{t_X\{p/X\}/Y, \widetilde{t_X'^Z}\{t_X\{p/X\} /X\} /\widetilde{Z}\}$.
  Since  $p \RS t_X\{p/X\}$, by Theorem~\ref{L:LLTS}(3) and  $s \equiv C_{Y,\widetilde{Z}}'\{t_X\{p/X\}/Y, \widetilde{t_X'^Z\{p /X\}} /\widetilde{Z}\} \notin F$, we get
  $s \RS w $
  and hence $ w \notin F$.
  So $u \notin F$ by Lemma~\ref{L:precongruence_F}.
  Since $C_Y\{\langle X|X=t_X \rangle /Y\} \sactions{\epsilon}u$, by Theorem~\ref{L:LLTS}(2), we get $C_Y\{\langle X|X=t_X\rangle/Y\} \fsactions{\epsilon} u$.
  Moreover, since $X$ is strongly guarded in $t_X'^Z$ for $Z \in \widetilde{Z}$, $X$ is strongly guarded in $C_{Y,\widetilde{Z}}'\{t_X\{p/X\}/Y, \widetilde{t_X'^Z}/\widetilde{Z}\}$.
  So $w \equiv C_{Y,\widetilde{Z}}'\{t_X\{p/X\}/Y, \widetilde{t_X'^Z} /\widetilde{Z}\}\{t_X \{p/X\}/X\} \not\action{\tau}$ due to Lemma~\ref{L:ONE_ACTION_TAU}(3) and $s \equiv C_{Y,\widetilde{Z}}'\{t_X\{p/X\}/Y, \widetilde{t_X'^Z} /\widetilde{Z}\}\{p/X\} \not\action{\tau}$.
  Thus $s \SRS w{\mathcal R} u$  because of $s \RS w $.

  \textbf{(Upto-2)} Let $C_Y\{t_X\{p/X\}/Y\}$ and $C_Y\{\langle X|X=t_X\rangle /Y\}$ be stable, and let $C_Y\{t_X\{p/X\}/Y\} \fsactions{a} s$.
  Then $C_Y\{t_X\{p/X\}/Y\} \faction{a} r\fsactions{\epsilon} s$ for some $r$.
  For the transition $C_Y\{t_X\{p/X\}/Y\} \action{a} r$, there exists $C_{Y,\widetilde{Z}}'$ that satisfies clauses (2.1) and (2.2) in Lemma~\ref{L:ONE_ACTION_TAU}.
   Then $r \equiv C_{Y,\widetilde{Z}}'\{t_X\{p/X\}/Y,\widetilde{r_Z'}/\widetilde{Z}\}$ for some $\widetilde{r_Z'}$  such that $t_X\{p/X\} \action{a} r_Z'$ for $Z \in \widetilde{Z}$.
   For each such transition, say $t_X\{p/X\} \action{a} r_Z'$, %since $X$ is strongly guarded in $t_X$,
    by Lemma~\ref{L:ONE_ACTION_TAU}(3), $r_Z' \equiv t_X'^Z\{p/X\}$ and $t_X\{\langle X|X=t_X\rangle /X\} \action{a} t_X'^Z\{\langle X|X=t_X \rangle /X\}$ for some $t_X'^Z$.
  Then, by $Ra_{16}$, $ \langle X|X=t_X\rangle   \action{a} t_X'^Z\{\langle X|X=t_X \rangle /X\}$ for $Z \in \widetilde{Z}$.
 % Moreover, by (\ref{L:GREATEST_SOLUTION}.1) with $D_Y \equiv C_Y$ and $a \in Act$, $C_Y\{\langle X|X=t_X \rangle/Y\}$ is stable.
  Further, since %$C_Y\{\langle X|X=t_X \rangle/Y\}$ is stable and
  $C_{Y,\widetilde{Z}}'$  satisfies clause (2.2) in Lemma~\ref{L:ONE_ACTION_TAU}, we get $C_Y\{\langle X|X=t_X \rangle/Y\} \action{a}v$, where
 $
  v \stackrel{\vartriangle}{\equiv} C_{Y,\widetilde{Z}}'\{\langle X|X=t_X \rangle /Y, \widetilde{t_X'^Z}\{\langle X|X=t_X \rangle /X\}/\widetilde{Z} \}$.
  Let $u \stackrel{\vartriangle}{\equiv} C_{Y,\widetilde{Z}}'\{t_X\{p/X\} /Y, \widetilde{t_X'^Z}\{t_X\{p/X\} /X\}/\widetilde{Z} \}$.
  By Theorem~\ref{L:LLTS}(3), we have
 $r \equiv C_{Y,\widetilde{Z}}'\{t_X\{p/X\} /Y, \widetilde{t_X'^Z\{p /X\}}/\widetilde{Z} \} \RS u$ because of $p \RS t_X\{p/X\}$.
  Further, it follows from $r \fsactions{\epsilon} s$ that $u \fsactions{\epsilon} t$ and $s \SRS t$ for some $t$.
  Since $u {\mathcal R} v$, by (Upto-1), there exits $t'$ such that $v \fsactions{\epsilon} t'$ and $t \SRS {\mathcal R} \SRS t'$.
  Moreover,  by Lemma~\ref{L:precongruence_F}, $C_Y\{\langle X|X=t_X \rangle /Y\}\notin F$  due to  $C_Y\{t_X\{p/X\} /Y\}\notin F$.
  Hence $C_Y\{\langle X|X=t_X \rangle /Y\}\faction{a}v \fsactions{\epsilon} t'$ and $s \SRS t \SRS {\mathcal R} \SRS t'$, as desired.
 %\textbf{(ALT-upto-3)} It is implied by (\ref{L:GREATEST_SOLUTION}.1).
\end{proof}

%Now we arrive at the main result of this note.
As a consequence of
%It is easy to see that $\langle X|X=t_X \rangle$ is a solution of the equation $X=_{RS}t_X$.
Lemma~\ref{L:GREATEST_SOLUTION} and Theorem~\ref{L:LLTS}(3), our main result is arrived, which characterizes $\langle X|X=t_X \rangle$ as the greatest solution of $X=_{RS}t_X$.
\begin{theorem}\label{T:GREATEST_SOLUTION}
For any $t_X$ with strongly guarded $X$,
  %For any equation $X=_{RS}t_X$ such that $X$ is strongly guarded in $t_X$,
  $\langle X|X=t_X\rangle$ is the greatest solution (w.r.t $\RS$) of $X=_{RS}t_X$; moreover $\langle X|X=t_X\rangle$ is consistent iff consistent solutions exit.
\end{theorem}
%\begin{proof}
%It is easy to see that $\langle X|X=t_X \rangle$ is a solution of the equation $X=_{RS}t_X$.
%Moreover, since $X$ is strongly guarded in $t_X$, by Lemma~\ref{L:GREATEST_SOLUTION}, $\langle X|X=t_X\rangle$ is the greatest solution of the equation $X=_{RS}t_X$.
%Finally, if $p=_{RS}t_X\{p/X\}$ for some $p \notin F$ then by Theorem~\ref{L:LLTS}(3), we get $\langle X|X=t_X \rangle \notin F$.
%\end{proof}

We give a brief discussion to conclude this note.
For Theorem~\ref{T:GREATEST_SOLUTION}, the hypothesis that $X$ is strongly guarded cannot be relaxed to that $X$ is weakly guarded.
For instance,
consider the equation $X=_{RS}\tau.X$, since $p=_{RS}\tau.p$ always holds for any $p$, such equation has infinitely many consistent solutions. However, since $\langle X|X=\tau.X \rangle$ is inconsistent by Theorem~\ref{L:LLTS}(1) and (LTS2) in Definition~\ref{D:LLTS}, it is the least solution of the equation $X=_{RS}\tau.X$.
%\section{Conclusions and discussion}
%
%  This paper works on LLTS-oriented process calculus $\text{CLL}_R$ furtherly.
%  We show that for any given equation $X=_{RS}t$ such that $X$ is strongly guarded in $t$, $\langle X|X=t \rangle$ is the greatest consistent solution w.r.t $\RS$ if consistent solutions exist.
%
%  For further work, we have two directions.
%  Firstly, we will clear the structure of the solution space $\{p:p \RS t_X\{p/X\}\}$ where $X$ is strongly guarded in $t_X$.
%  Secondly, it is very interesting to encode action-based CTL \cite{Luttgen11,Nicola90} in $\text{CLL}_R$ and study the relationship between this encoding and classical model checking.\\

 $\; $

\noindent \textbf{References}

    [1] {R.~Bol},
    {J.F.~Groote},
    {The meaning of negative premises in transition system specifications},
    {JACM}
    {43}
  ({1996})
    {863-914}.

    [2] {G.~L\"{u}ttgen},
    {W.~Vogler},
    {Conjunction on processes: full-abstraction via ready-tree semantics},
    {TCS}
    {373}
  ({1-2})
  ({2007})
    {19-40}.

   [3] {G.~L\"{u}ttgen},
    {W.~Vogler},
    {Ready simulation for concurrency: it's logical},
    {Inform. \& comp.}
    {208}
  ({2010})
    {845-867}.

   [4] {G.~L\"{u}ttgen},
    {W.~Vogler},
    {Safe reasoning with Logic LTS},
    {TCS}
    {412}
  ({2011})
    {3337-3357}.

   [5] {Y.~Zhang, Z.H.~Zhu, J.J. Zhang},
    {On recursive operations over logic LTS},
    {MSCS, available on CJO 2014 doi:10.1017/S0960129514000073.}

\end{document}